\documentclass[twoside]{aiml22}

\usepackage{aiml22macro}

\usepackage{graphicx}
\usepackage{amsmath}
\usepackage{amssymb}
\usepackage{comment}
\usepackage{changebar}

\usepackage[shortlabels]{enumitem}
\usepackage{mathtools}
\usepackage{algpseudocode}

\usepackage{color}

\def\lastname{Li, Belardinelli}

\begin{document}

\begin{frontmatter}
  \title{A Sahlqvist-style Correspondence Theorem\\ for Linear-time Temporal Logic}
  \author{Rui Li}\footnote{
  %\cbstart 
  Undergraduate Research Opportunities Program (UROP) fund, Department of Computing, Imperial College London, summer 2021}
  %\cbend}
  \address{Sorbonne Université \\ Paris, France}
  \author{Francesco Belardinelli}
  \address{Department of Computing\\ Imperial College London\\ United Kingdom}

  \begin{abstract}
The language of modal logic is capable of expressing first-order conditions on Kripke frames. For instance, the modal formula $(\Box q\to q)$ is valid in exactly the reflexive frames, where reflexivity $\forall x R(x,x)$ is a first-order condition. The classic result by Henrik Sahlqvist identifies a significant class of modal formulas for which first-order conditions -- or Sahlqvist correspondents --  can be find in an effective, algorithmic way. Recent works have successfully extended this classic result to more complex modal languages. In this paper, we pursue a similar line and develop a Sahlqvist-style correspondence theorem for Linear-time Temporal Logic (LTL), which is one of the most widely used formal languages for temporal specification. LTL extends the syntax of basic modal logic with dedicated temporal operators Next $X$ and Until $U$. As a result, the complexity of the class of formulas that have first-order correspondents also increases accordingly. In this paper, we identify a significant class of LTL Sahlqvist formulas built by using modal operators $F$, $G$, $X$, and $U$. The main result of this paper is to prove the correspondence of LTL Sahlqvist formulas to frame conditions that are definable in first-order language. 
  \end{abstract}

 \end{frontmatter}

%\maketitle

\section{Introduction}

One of the most well-known results in the model theory of modal logic is that modal languages are rich enough to express 
%The language of modal logic is capable of expressing 
(first-order) conditions on Kripke frames. Results along this direction have been known as {\em Correspondence Theory} \cite{Benthem84,BlackburnRV02}. For instance, the modal formula $(\Box q\to q)$ is valid in exactly the reflexive frames, where reflexivity $\forall x R(x,x)$ is a first-order condition. 
Since the 1970's, much research in modal logic has been devoted to identifying classes of formulas for which such first-order correspondents exist, including algorithms for their automatic computation.
The classic result by H.~Sahlqvist \cite{Sahlqvist75} identifies a significant class of modal formulas for which first-order conditions -- or Sahlqvist correspondents --  can be found in an effective, algorithmic way. Since then, correspondence theory has been successfully extended to more complex and expressive modal languages \cite{GabbayHR94,SambinV89,BenthemBH12}.
\paragraph{Contribution.}
In this paper we 
%pursue a similar line and 
develop a Sahlqvist-style correspondence theorem for Linear-time Temporal Logic (LTL), which is nowadays one of the most widely-used formal languages for temporal specification \cite{BaierK08}. LTL extends the syntax of basic modal logic with dedicated temporal operators {\em Next} $X$ and {\em Until} $U$. 
Formulas in LTL are interpreted on infinite words -- or {\em paths} -- representing the execution of a reactive system.
Interestingly, Kamp \cite{Kamp68} proved that every temporal operator on a class of continuous, strict linear orderings that is definable in first-order logic is expressible in terms of {\em Since} $S$ and {\em Until} $U$.
As a result, the complexity of the class of modal formulas that have first-order correspondents also increases accordingly. In this paper, we identify a significant class of LTL formulas built by using temporal operators {\em Eventually} $F$, {\em Always} $G$, {\em Next} $X$, and {\em Until} $U$. 
 To accommodate the enhanced expressiveness, we extend the class of Sahlqvist formulas with some additional conditions. To facilitate the treatment, we introduce the ``intermediate" logic LTL', which is more expressive than LTL, but whose syntax is closer to that of normal modal logics.
Our main result 
%of this paper 
is to prove the correspondence of such Sahlqvist formulas in LTL to frame conditions that are definable in a first-order language. 

\paragraph{Related Work.}
As we mentioned above, Sahlqvist correspondence theorem has been extended in a number of different directions, mainly considering more and more expressive modal languages. For instance, in \cite{GabbayHR94} a correspondence theorem is proved for temporal modal logic,  whereas in \cite{BenthemBH12,BezhanishviliH12} similar results are proved for the $\mu$-calculus and modal fixed-point logic respectively. It has to be remarked that these works extend the proof given in \cite{SambinV89}, rather than Sahlqvist's original result in \cite{Sahlqvist75}.
More recently, correspondence results have been proved for hybrid logics \cite{ConradieR17}, distributive modal logics \cite{GehrkeNV05}, and polyadic modal logics \cite{GorankoV02}.
Some efforts have also been applied to the problem of finding more general and efficient algorithms to compute first-order correspondents of modal formulas \cite{ConradieGV06,ConradiePS17,ZhaoPC19}, including \cite{SambinV89} mentioned earlier.
Still, to the best of our knowledge, no comparable result has been proved for the kind of temporal logics used in the specification and verification of reactive and distributed systems \cite{BaierK08}. We deem such a result of interest to theoreticians and practitioners in modal logics alike.

%Sahlqvist Formula is the class of formulas introduce by Sahlqvist for modal logics with only the conventional modal operators $\Box$ and $\Diamond$ [ ]. The standard translations of these formulas are equivalent to some first order formula; subsequently, their frame conditions can be captured by first order logic. Linear Temporal logic is the extension of normal logic with additional operator $X$ and $U$ [ ]. The interpretation of this logic is ``discrete time'': a linear order $L$ on which if $x,y\in L$ and $x< y$, then the successor of $S(x)\leqslant y$. In LTL, $F$ and $G$ are the equivalents of $\Diamond$ and $\Box$ with meaning ``for some time in the future'' and ``always in the future''. However, LTL is strictly more expressive than normal modal logic with $X\phi$ and $\phi U\varphi$ respectively interpreted as ``$\phi$ happens at the next point in time'' and ``$\phi$ takes place until $\varphi$ takes place''. 

\paragraph{Structure of the Paper.}
%The structrucorganization of the paper is as follows: 
In Sec.~\ref{preliminaries} we introduce the syntax and semantics of LTL as well as the auxiliary logic LTL', and define correspondence between modal and first-order formulas. In Sec.~\ref{sec_Sahlqvist} we define the class of Sahlqvist formulas for LTL and LTL', and provide a few preliminary results. Finally, Sec.~\ref{correspondence} is devoted to the main result of this paper, namely the proof of the correspondence theorem. 

\section{Preliminaries: Linear-time Temporal Logic} \label{preliminaries}

In this section we provide background information about  
%the syntax and semantics of 
Linear-time Temporal Logic (LTL) \cite{BaierK08,HuthR00}.
Specifically, in Sec.~\ref{syntax} we introduce its syntax, as well as the syntax of the auxiliary language LTL'. Then, in Sec.~\ref{semantics} we interpret both languages on infinite system executions.
Finally, in Sec.~\ref{sec_translation} we define their standard translations \cite{BlackburnRV02}.

\subsection{LTL: Syntax} \label{syntax}

We fix a set $Prop$ of atomic propositions (or {\em atoms}) and define the formulas $\phi$ in Linear-time Temporal Logic in Backus-Naur form as follows: 
%Let's first state the syntax of LTL: 
\begin{eqnarray*}
  \phi  & = &   Prop \mid \bot \mid \top \mid \neg \phi \mid \phi \wedge \phi \mid \phi \vee \phi \mid G \phi \mid F \phi \mid X\phi \mid \phi U \phi
\end{eqnarray*}
where $G$ is read ``{\em always}", $F$ ``{\em eventually}", $X$ is the {\em Next} operator, and $U$ is the {\em Until} operator \cite{BaierK08}. The Boolean connectives $\to$ and $\leftrightarrow$ can be introduced as standard.
Operators $F$ and $G$ can be defined in terms of $U$, but for convenience we assume them as primitive.

In this paper, we consider also a variant of LTL, that we call LTL'. %\cbstart 
Let $W$ be a set of possible worlds (which serves as the model of LTL and LTL'). Fix $W$, then we define the syntax of LTL' w.r.t.~this particular $W$ as follows: 
\begin{eqnarray*}
  \phi  & =  &  Prop \mid \bot \mid \top \mid \neg \phi \mid \phi \wedge \phi \mid \phi \vee \phi \mid G \phi \mid F_x \phi \mid  \widehat{G}_{w,w'} \phi \mid X \phi
\end{eqnarray*}
where $w$ and $w'$ ($w\neq w'$) are states in $W$, and $x$ is a variable over states. 
%\cbend 
We will also use the following convention: $w<w'$. If it were the case that $w'<w$, then it suffices to switch the place of $w$ and $w'$ and write $\widehat{G}_{w,w'}$ as $\widehat{G}_{w',w}$. Remark that $\widehat{G}_{w,w}$ is not in the language of LTL'.
\begin{remark}
    Although states are semantical notions, the symbols representing them can be treated syntactically. The difference between the use of symbols in $F_x$ and $\widehat{G}_{w,w'}$ is that $x$ in the former is a variable that ranges over possible states, of which $w$ and $w'$ in the latter are members. In this paper, $x, y, z,\ldots$ would be used to denote variables, whereas $w, w', s, s',\ldots$ would denote states that are fixed in the context. Also, $u,u'$ and $v,v'$ can be used interchangeably, whenever the context is clear.
%\end{remark}
%\begin{remark}
    %In particular, if the model of LTL' is finite, then $G$ can be expressed via $\widehat{G}_{w,v}$ for some $w, v$ that are fixed in the context. 
\end{remark}

\subsection{LTL: Semantics} \label{semantics}

%\FB{We need to add definition of $w$, $h$, etc.}

%\FB{Are we interpreting formulas on worlds, or on paths (i.e., sequences of worlds)?}

To provide a semantics to LTL, we consider {\em transition systems} $T = (S, \to)$, where $S$ is a set of {\em states}, and the {\em transition relation} $\to\ \subseteq S \times S$ is a binary relation on $S$. Normally, the relation $\to$ is assumed to be {\em serial}: for all $s \in S$, there exists $s' \in S$ such that $s \to s'$.  
%We omit the labelling function here as the frame-theoretic model of LTL' can be based on any arbitrary transition system independent of the interpretation of states that we give to it. 
Then, a {\em path} in a transition system is an infinite sequence $s_1, s_2, s_3,\ldots$, where for all $i \in \mathbb{N}$, $s_i \to s_{i+1}$. 

We now define models for LTL. Let $W$ be the set of all 
%possible 
paths in $T$; whereas $\leqslant$, $<$, and $\mathbf{S}$ are all binary relations on $W$, introduced as follows. Let $w = s_1, s_2, s_3, \ldots$ and $v = s'_1, s'_2, s'_3, \ldots$ be 
%two arbitrary 
paths in $T$, then $w \leqslant v$ iff for some $i \geqslant 1$, $s_i = s_1'$ and for all $j > 0$, $s_{i+j} = s'_{1+j}$, that is, $v$ is a {\em subpath} of $w$ starting from some index $i$. Then, $w < v$ iff $w \leqslant v$ and $w \neq v$.
%similarly, $w \leqslant v$ iff for some $i \geqslant 1$, $s_i = s_1'$ and for all $j > 0$, $s_{i+j} = s'_{1+j}$; 
Further,  $\mathbf{S}$ means {\em successor}: $v = \mathbf{S}(w)$ iff for all $i > 0$, $s'_{i} = s_{i+1}$. When the context is clear, we sometimes simply write $R(w,v)$ for $w < v$, $w \leqslant v$ or $v = \mathbf{S}(w)$. 

A {\em model} for LTL is a tuple $M = (T, h)$, where $T$ is a transition system, and $h: Prop \to 2^{S}$ is an {\em assignment function} from atoms to set of states in $S$.
We lift the assignment $h$ from states to paths so that $w \in h(q)$ iff $s_1 \in h(q)$. 
%LTL' formulas to $W$. 
%
%\newtheorem{definition}{Definition}
\begin{definition}[Satisfaction]
Given a model $M$, path $w$, and formula $\phi$ in LTL', the {\em satisfaction relation} $\vDash$ is defined as follows:
\begin{tabbing}
    $(M,w) \vDash q$ \ \ \ \ \ \ \ \ \= iff \ \ \= $w \in h(q)$\\
    $(M,w) \vDash \neg\phi$ \> iff \> $(M,w) \not \vDash \phi$\\
    
    $(M,w) \vDash \phi \wedge \varphi$ \> iff \>  $(M,w) \vDash \phi$ and $(M,w) \vDash \varphi$\\
    
    $(M,w) \vDash \phi \vee \varphi$ \> iff \>  $(M,w) \vDash \phi$ or $(M,w) \vDash \varphi$\\
    
    $(M,w) \vDash G \phi$ \> iff \> for all
    $v \in W$, $w \leqslant v$ implies $(M,v) \vDash \phi$\\

    $(M,w) \vDash F_x \phi$ \> iff \> for some 
    $x \in W$, $w \leqslant x$ and $(M, x) \vDash \phi$ \\

    $(M,w) \vDash \widehat{G}_{w, w'} \phi$ \> iff \>
    for all $u \in W, w \leqslant u < w'$ implies $(M,u) \vDash \phi$\\

    $(M,w) \vDash X \phi$ \> iff \> 
    $v = \mathbf{S}(w)$ and $(M, v) \vDash \phi$ 
\end{tabbing}
\end{definition}

Hereafter we use $w \vDash \phi$ as an abbreviation for $(M, w) \vDash \phi$.
We write $[\phi]^h_w=1$ iff $(M, w) \vDash \phi$ for $M = (T, h)$. For future references, we precisely define below assignments for arbitrary formulas. 
\begin{definition}[Assignment]
%    \cbstart 
    Let $T=(S,\to)$ be a transition system, and $h: Prop \to 2^{W}$ an {\em assignment function} as before. We extend the domain of $h$ from the set of atoms $Prop$ to the set  $Form$ of all formulas: 
    $$h: Form \to 2^W$$
%\cbend
    such that 
    %whenever $\phi$ is not an atom, then 
    $h(\phi)$ is defined as
    $\{w\in S \mid (T,h,w) \vDash\phi\}$.
\end{definition}
%\vspace{0.3cm}

To provide an interpretation for LTL, we replace the clause for $\widehat{G}_{w,w'}$ with a clause for the Until operator $U$, as follows:
%Therefore, the only place where the semantics of LTL differs from that of LTL' is the following clause: 
%
\begin{tabbing}
    $(M, w) \vDash \phi U \phi'$  \ \  \= iff \ \ \= for some  $u \geqslant w$, $(M, u) \vDash \phi'$, and\\
    \> \> for all $v \in W$, $w \leqslant v < u$ implies $(M, v) \vDash \phi$
\end{tabbing}

LTL also replaces $F_x$ with the operator $F$, where the variable path is no longer shown in the syntax. But its semantics remains the same. 
\vspace{0.2cm}

Now it is possible to translate LTL into LTL'. 
\begin{definition} \label{translation}
    Let $Form_{LTL}$ be the class of all LTL formulas and $Form_{LTL'}$ be the class of all LTL' formulas. Let
    %\fb{since $t$ is also used later on to denote terms, we can denote the translation as $\tau$.}
    \begin{align*} 
        \tau: Form_{LTL}& \to Form_{LTL'}
    \end{align*}
    be the translation from LTL to LTL' defined as follows:
\begin{tabbing}
    $q$ \ \ \ \ \ \ \ \ \ \= $\mapsto$ \= $q$\\
    $\neg \phi$ \> $\mapsto$ \> $\neg \tau(\phi)$\\
    $\phi_1 \wedge \phi_2$ \> $\mapsto$ \>  $\tau(\phi_1) \wedge \tau(\phi_2)$\\
    $\phi_1 \vee \phi_2$ \> $\mapsto$  \>  $\tau(\phi_1) \vee \tau(\phi_2)$\\
    $G \phi$ \> $\mapsto$ \> $G\tau(\phi)$\\
    $F \phi$ \> $\mapsto$ \> $F_x\tau(\phi)$\\
    $X \phi$ \> $\mapsto$ \> $X \tau(\phi)$\\
    $\phi_1 U\phi_2$ \> $\mapsto$ \> $F_x (\tau(\phi_2) \land \widehat{G}_{w,x} \tau(\phi_1))$ 
\end{tabbing}
where $x$ is a path variable, and $w$ is the path at which we aim to evaluate the formula.
\end{definition}
\begin{remark}[Variable Convention]

    In the conjunctive and disjunctive clause, if a path variable $x$ appears in both $\tau(\phi_1)$ and $\tau(\phi_2)$, then in $\tau(\phi_1\land\phi_2)$ we replace $x$ in $\tau(\phi_2)$ by another path variable $x'$ that do occur in either $\tau(\phi_1)$ or $\tau(\phi_2)$. 
    
    If $x$ in $\tau(\phi_1 U \phi_2)$ appears in $\tau(\phi_1)$ or $\tau(\phi_2)$, then we replace the occurrences of $x$ in $\tau(\phi_1)$ and $\tau(\phi_2)$ by $x_1$ and $x_2$. 
\end{remark}

%\fb{we need a lemma that states that truth is preserved by translation $t$. Otherwise it is not clear what is the purpose of the translation.}
%\cbstart
\begin{lemma} \label{translation_lemma}
    Let $\tau$ be the translation from LTL to LTL' in Def.~\ref{translation}. Then an LTL formula and its translation w.r.t.~$\tau$ are semantically equivalent.
\end{lemma}
%\begin{proof}
%The proof makes use of structural induction on the formula. We only consider the case for the LTL formula $\phi = \beta U \psi$, where $\psi$ is also LTL untied \fb{untied formulas are not introduced yet, we have to prove the lemma for generic $\phi_1$ and $\phi_2$}. Let $w$ be any path, and $\phi$ is evaluated at $w$. The standard translation of $\tau(\phi)$ at $w$ is
    %$$ST_w(F_x (\tau(\psi) \land \widehat{G}_{w,x} \tau(\beta))) = $$ 
    %$$\exists x, (w<x \land ST_x(\tau(\psi)) \land (\forall u, w\leqslant u < x\to ST_u(\tau(\beta))))$$
    %By induction hypothesis, $ST_x(\tau(\psi))\Leftrightarrow ST_x(\psi)$ and $ST_u(\tau(\beta))\Leftrightarrow ST_u(\beta)$. So $ST_w(\tau(\phi))\Leftrightarrow ST_w(\phi)$ \fb{the standard translation has not been introduced yet, we better prove this by using the semantics}.
    %By Lemma 2.8, that means $\phi$ and $\tau(\phi)$ are both satisfied or both unsatisfied at $w$. Since $w$ is arbitrary, $\phi$ and $\tau(\phi)$ are semantically equivalent. 
%\end{proof}
%/\cbend
%\cbstart
\begin{proof}
The proof makes use of structural induction on the formula. We only consider the case for the LTL formula $\phi = \phi_1 U \phi_2$. Let $w$ be any path, and $\phi$ is evaluated at $w$. The translation $\tau(\phi)$ of $\phi$ at $w$ is
    %$$ST_w(F_x (\tau(\psi) \land \widehat{G}_{w,x} \tau(\beta))) = $$ 
    $$\exists x (w<x \land \tau(\phi_2) \land \forall u (w\leqslant u < x\to \tau(\phi_1)))$$
    By induction hypothesis, $x\vDash \tau(\phi_2) \Leftrightarrow x\vDash \phi_2$ and $u\vDash \tau(\phi_1) \Leftrightarrow u\vDash \phi_1$. So $w\vDash \tau(\phi) \Leftrightarrow w\vDash \phi$. Since $w$ is arbitrary, $\phi$ and $\tau(\phi)$ are semantically equivalent. 
\end{proof}
%\cbend
%Notice that the translation from LTL into LTL' is possible by expressing $\phi U \psi$ as $F_x\psi\wedge G_x\phi$, where $x$ is a path variable. 

\subsection{Standard Translation and Correspondence} \label{sec_translation}

The standard translation of formulas in LTL' mirrors their semantics.
%, except that we treat the language of metasystem of LTL' as if it is an object system. 
For every atom $q \in Prop$, we introduce a predicate symbol $Q$. For an arbitrary formula $\phi$ in LTL', 
%say it is satisfied at $w$, 
we denote the first-order standard translation of $\phi$ at $w$ as $ST_w(\phi)$, and it is inductively defined as follows: 
\begin{definition}[Standard Translation] \label{sttranslation}
The standard translation $ST_w(\phi)$ of formula $\phi$ at path $w$ is inductively defined as case of $\phi$:
\begin{tabbing}
    $q$ \ \ \ \ \ \ \ \ \ \ \ \ \ \ \= : \ \ \= $Q(w)$\\
    $\neg \phi$ \> : \> $\neg ST_w(\phi)$\\
    $\phi \wedge \varphi$ \> : \>  $ST_w(\phi) \wedge ST_w(\varphi)$\\
    $\phi \vee \varphi$ \> :  \>  $ST_w(\phi) \vee ST_w(\varphi)$\\
    $G \phi$ \> : \> $\forall v (w \leqslant v \to ST_v(\phi))$\\
    $F_x \phi$ \> : \> $\exists x (w \leqslant x \wedge ST_x(\phi))$\\
    $\widehat{G}_{s,s'} \phi$ \> : \> $\forall v (s \leqslant v < s' \to ST_{v}(\phi))$\\
    $X \phi$ \> : \> $ST_{\mathbf{S}(w)}(\phi)$
\end{tabbing}
\end{definition}

To simplify the notation, instead of saying that $ST_w(\phi)$ is the standard translation of $\phi$ at $w$, we say that it is the standard translation of $\phi[w]$. Then, the second-order standard translation of $\phi[w]$ is obtained by prefixing universal quantification for every predicate $Q_1, Q_2, \ldots, Q_k$ in 
%front of 
$ST_w(\phi)$. There is no abbreviated notation for this second-order standard translation. Whenever the context is clear, we will also call it the standard translation. For the most part, we work with the second-order standard translation. %the reason for doing so will be clear by Def.~\ref{sttranslation}. 

Since the models for LTL' and for first-order logic are the same (they are both relational structure), we say that $(M, w)\vDash \phi$, where $\phi$ is a first order formula. However, when it comes to the second-order formulas, the models have to be modified. In second-order logic, quantification over predicates (sets) is allowed, and the domain of a predicate is determined by the assignment $h$, i.e.,
%\cbstart 
$dom(Q) = 
%\{w \in W \mid w \in h(q) \} = 
h(q)$.
%\cbend
Therefore, assignments in transition systems are equivalent to (universal) quantification over predicates in second-order logic. 

\begin{definition}[Correspondence]
    Let $T = (S, \to)$ be a transition system, and $w \in W$. An LTL' formula $\phi(q_1, q_2,\ldots, q_k)$ is said to {\em (locally) correspond} to a formula $\varphi$ in second order logic at $w$ whenever $\phi$ are $\varphi$ are 
%\cbstart 
    both evaluated to be true
%\cbend 
    at $w$ in $T$.
\end{definition}

The following lemma shows why local correspondence is defined the way it is. A proof can be obtain by a straightforward induction on the structure of formula $\phi$.  %\fb{since we have space, you might consider adding at least a sketch of proof.} \rui{ I actually don't think we have space, since there will be conclusion and examples to be added. Besides, this lemma is just too straightforward. I will add something for the next comment.}
%However, we do not present a proof. 
%
\begin{lemma}
An LTL' formula $\phi(q_1, q_2, \ldots, q_k)$ (locally) corresponds to $\forall Q_1 \forall Q_2 \ldots \forall Q_k ST_w(\phi)$ 
%\cbstart 
at $w$, 
%\cbend 
where $ST_w(\phi)$ is the (first-order) standard translation of $\phi[w]$.
%For any LTL' frame $\xi$ and any state in $\xi$, $(\xi, w) \vDash \phi$ iff $\xi \vDash \forall Q_1, \forall Q_2,\ldots, \forall Q_k, ST_w(\phi)$ where $ST_w(\phi)$ is the standard translation for $\phi(q_1, q_2,\ldots, q_k)[w]$. 
\end{lemma}
%\cbstart
\begin{remark}
    In light of this lemma, we will be using semantics and standard translation interchangeably in this paper. 
\end{remark}
%\cbend
%We can write $(h, w) \vDash \phi$ when for all models in which the assignment is $h$, $(M, w)\vDash \phi$. We say that $\phi$ locally corresponds to $ST(\phi)$ if 

The main result we prove in this paper can be stated informally as follows: there is a collection of LTL formulas $\phi$, such that for all paths $w$, the local correspondent $\phi[w]$ of $\phi$ at $w$ can be expressed as a first-order formula. This is the basic content of Sahlqvist correspondence theorem, which will be stated later on in more precise terms. Note that, although the standard translation is only defined for LTL', the translation for LTL and its semantics can be defined in a similar manner, where the main difference is the following clause: 
\begin{tabbing}
    $ST_w(\phi U \phi'$) \= = \= $\exists u, w\leqslant u \wedge ST_u(\psi)\wedge (\forall v, w\leqslant v < u\to ST_v(\phi'))$
\end{tabbing}

\section{Sahlqvist Formulas for LTL} \label{sec_Sahlqvist}

In this section, we introduce two particular types of formulas that play key roles in the construction of Sahlqvist formulas: boxed formulas and negative formulas. We prove the monotonicity theorem in Sec.~\ref{negative} and introduce Sahlqvist formulas for LTL in Sec.~\ref{Sahlqvist}.

\subsection{Boxed Formulas}

In standard modal logic, boxed formulas are defined as a sequence of boxes $\Box$ followed by an atomic formula, i.e., they have the form $\Box \ldots \Box q$ for a possible empty sequence of boxes. The sequence of boxes can be denoted as $\Box^n$, for $n \in \mathbb{N}$, whose semantics is similar to the one for a single box: $w \vDash \Box^{n} q$ iff for all $v$, $R^n(w, v)$ implies $v \vDash q$, where $R^n$ is not difficult to construct (see Lemma~\ref{boxed}).

Similarly for LTL', the syntactic operators having universally quantified implication as semantics can be integrated into the LTL' boxed formulas for the same reason. 
%\cbstart 
We denote an arbitrary boxed formula as $\boxplus^{n} q=\boxplus\ldots\boxplus q$, where each $\boxplus$ is a distinct element from $\{G, \widehat{G}_{w,w'}, X\}$ (i.e., the set of boxed operators). Now we define the corresponding accessibility relation. 

\begin{definition}[Accessibility Relation $R_{\boxplus^n}$] \label{accessibility}
We define the accessibility relation $R_{\boxplus^n}$ by induction on $n\in \mathbb{N}$.

Base case: if $n = 0$, i.e.  $\boxplus^n q = q$, then  $R^0(w,v)$ iff $w = v$.

Inductive cases: let $R_{\boxplus^n}$ be defined, then 
\begin{itemize}
\item If $\boxplus^{n+1} q = G \boxplus^{n}q$, then $R_{\boxplus^{n+1}}(w,v)$ iff for some $u \in W$, $w \leqslant u$ and $R_{\boxplus^n}(u,v)$. 

\item If $\boxplus^{n+1} q = X \boxplus^{n}q$, then $R_{\boxplus^{n+1}}(w,v)$ iff $R_{\boxplus^n}(\mathbf{S}(w),v)$.

\item If $\boxplus^{n+1} q = \widehat{G}_{s,s'} \boxplus^{n}q$, then $R_{\boxplus^{n+1}}(w,v)$ iff for some $u \in W$, $s \leqslant u < s'$ and $R_{\boxplus^n}(u,v)$. 
\end{itemize}
Whenever the context is clear, we use $R^n$ to denote $R_{\boxplus^n}$. %\cbend
\end{definition}

By Def.~\ref{accessibility} we can prove the following auxiliary result concerning boxed formulas.
\begin{lemma}[Boxed Formulas Lemma] \label{boxed}
    Let $\boxplus^n q$ be an LTL' boxed formula with $n$ boxed operators appearing in front of atom $q$ (with $n$ possibly equal to 0). 
    Then 
    %there exists a relation $R$ such that 
    $w\vDash \boxplus^n q$ iff for all $v \in W$, $R^n(w,v)$ implies $v \vDash q$.
\end{lemma}
\begin{proof}
    We prove this lemma by induction on $n$. The base case for $n = 0$ is immediate: $w \vDash q$ iff for all $v$, $w=v$ implies $v\vDash q$, that is $R^0(w,v)$. Now suppose that the lemma holds for an arbitrary $n$, i.e., $w\vDash \boxplus^n q$ iff for all $v \in W$, $R^n(w,v)$ implies $v \vDash q$. We have to show that $w\vDash \boxplus^{n+1} q$ iff for all $v \in W$, $R^{n+1}(w,v)$ implies $v \vDash q$. We discuss by case the options for the first boxed operator $\boxplus_0$ in $\boxplus^{n+1} q$. 
    
    For 
    %the case where 
    $\boxplus_0 = \widehat{G}_{s,s'}$, $w \vDash \boxplus_0\boxplus^n q$ iff for every $v$, $s\leqslant v< s'$ implies $v\vDash \boxplus^n q$. By induction hypothesis, $w\vDash \boxplus_0\boxplus^n q$ iff for every $v$, $s\leqslant v<s'$ implies that for every $u$ such that $R^n(v,u)$, $u\vDash q$. We want to show that this is equivalent to: for all $u$, $R^{n+1}(w,u)$ implies $u\vDash q$. Suppose $w\vDash \boxplus_0\boxplus^n q$ is the case; fix $u_0$. Let $v_0$ be a path such that $s\leqslant v_0 < s'$, if $R^n(v_0,u_0)$, then $u_0\vDash q$ by assumption. Since $u_0$ is arbitrary, we get $R^{n+1}(w,u_0)$ for every $u_0$. Conversely, assume that for every $u$, if for any $v$ such that $s\leqslant v< s'$ and $R^n(v,u)$, then $u\vDash q$. Fix $v_0$ such that $s\leqslant v_0< s'$. Then take an arbitrary $u_0$. If $R^n(v_0,u_0)$, then by assumption, $u_0\vDash q$. Since $v_0$ is arbitrary, for every $v_0$, $s\leqslant v_0< s'$ and $R^n(v_0,u_0)$ imply $u_0\vDash q$ for arbitrary $u_0$, as desired. This concludes the case for $\boxplus_0=\widehat{G}_{s,s'}$. 
    The proofs for the cases $\boxplus_0= G$ and $\boxplus_0= X$ are similar.
\end{proof}
\begin{remark}
    This lemma shows that the standard translation of every boxed formula $\boxplus^n q$ can be written in the form of $\forall v, R(w,v)\to Q(v)$ using a unique relation $R$. This construction will be invaluable in defining the minimal assignment for Sahlqvist formulas. 
\end{remark}

\subsection{Negative Formulas} \label{negative}

Similarly to standard modal logic, LTL' {\em positive formulas} $\phi$ can be defined as the ones constructed from atoms using $\wedge$, $\vee$, $G$, $F_x$, $\widehat{G}_{w,w'}$, $X$ only:
\begin{align*}
  \phi  =    Prop \mid \bot \mid \top \mid \phi \wedge \phi \mid \phi \vee \phi \mid G \phi \mid F_x \phi \mid \widehat{G}_{w,w'} \phi \mid X\phi
\end{align*}

An LTL' {\em negative formula} has one of the two following forms: 
\begin{enumerate}
    \item $\neg \phi$, where $\phi$ is an LTL' positive formula; 
    \item $\widehat{G}_{w,w'} N$, where $N$ is an LTL' negative formula.
\end{enumerate}

For example, $G F_x q$ is a positive formula, $\widehat{G}_{w,w'}\neg (q_1 \wedge q_2)$ is a negative formula; whereas $F_x(\widehat{G}_{w,x} q_1 \wedge \neg q_2)$ and $\neg X \neg X q$ are neither positive nor negative.

\begin{remark}
    Although negative formulas are defined to be a syntactic notion, the proof of the correspondence theorem is semantical. Therefore, whenever possible, if a formula is semantically equivalent to a negative formula, then we shall also call this formula negative. For example, if $N$ is a negative formula, then $GN$ is also a negative formula. 
    %\fb{we need to be careful with something like that.} \rui{ Yeah, I know, I will think about how to make it precise.}
\end{remark}

\begin{lemma}[Monotonocity] \label{monotonicity}
    Let $\phi$ be an LTL' positive formula, $q_1, \ldots, q_k$ be the atoms appearing in $\phi$, and $h_1$ and $h_2$ be assignments. 
    
    If for all $q_j$, $h_1(q_j) \subseteq h_2(q_j)$, then $h_1(\phi) \subseteq h_2(\phi)$. %\fb{did we introduce the notation $h_1(\phi)$, $h_2(\phi)$?}
\end{lemma}
\begin{proof}
    The proof is by induction on the structure of $\phi$. The base case for $\phi = q$ is immediate.
    %, and there is nothing to prove. 
    Now suppose that $\phi_1(q_1, \ldots , q_k)$ and $\phi_2(q_1,  \ldots, q_k)$ are two LTL' positive formulas that satisfy the statement of the lemma. We show that $\phi_3(q_1, \ldots, q_k)$, which is built from $\phi_1$ and $\phi_2$ using one of the operators from $\{\wedge, \vee, G, F_x, \widehat{G}_{w,w'}, X\}$ also satisfies the statement. If either $\phi_3 = G \phi_1$, or $\phi_3 = \widehat{G}_{w,w'} \phi_1$, or $\phi_3 = F_x \phi_1$ or $\phi_3 = X \phi_1$, since $h_1(\phi_1) \subseteq h_2(\phi_1)$, it is easy to see that $h_1(\phi_3) \subseteq h_2(\phi_3)$. It is also immediate to check that if $h_1(\phi_1) \subseteq h_2(\phi_1)$ and $h_1(\phi_2) \subseteq h_2(\phi_2)$, then $h_1(\phi_1 \wedge \phi_2) \subseteq h_2(\phi_1 \wedge \phi_2)$ and $h_1(\phi_1 \vee \phi_2) \subseteq h_2(\phi_1 \vee \phi_2)$. This concludes the induction. 
\end{proof}

\begin{corollary} \label{conversemonotonicity}
    Let $N$ be an arbitrary LTL' negative formula, $q_1, q_2,\ldots, q_k$ are the atomic variables appearing in $N$. Let $h_1$ and $h_2$ be two random assignments. If for all $q_j$, $h_1(q_j) \subseteq h_2(q_j)$, then $h_2(N) \subseteq h_1(N)$.
\end{corollary}

\begin{proof}
    If $N$ is of the first form of LTL' negative formula, then the statement follows directly from the Lemma~\ref{monotonicity}. Now suppose $N$ is of the second form, that is, $N = \widehat{G}_{s,s'} N'$, where $N'$ is negative. Let $w$ be any path, then $ST_w(N)$ is 
    $$\neg \exists y (s \leqslant y < s' \wedge ST_y(\neg N'))$$
    Remark that the part in the scope of the negation is in fact a positive fragment in the interpretation of LTL' formulas. By monotonocity lemma, if for all $q_j$ occurring in $N'$, $h_1(q_j) \subseteq h_2(q_j)$, then $h_1(\neg N') \subseteq h_2(\neg N')$. In other words, if $(T,h_1,y)\vDash ST_y(\neg N')$, then $(T,h_2,y)\vDash ST_y(\neg N')$. Therefore, if there is a path $y$ between $s$ and $s'$ such that $y\vDash ST_y(\neg N')$ and $y\in h_1(\neg N')$, then there is also such a path for $h_2$. So if $(T,h_1,x)\vDash \exists y (s \leqslant y < s' \wedge ST_y(\neg N'))$, then $(T,h_2,x)\vDash \exists y (s \leqslant y < s' \wedge ST_y(\neg N'))$. It follows that $h_2(N) \subseteq h_1(N)$, as desired.
\end{proof}

\subsection{Sahlqvist Formulas} \label{Sahlqvist}

The main goal of this paper is to find a significant class of Sahlqvist formulas for LTL, we therefore define them here. Then, we will show that this construction can be simplified by using the auxiliary language LTL'.

%Of course, instead of using boxed and negative formulas for LTL', we shall use the ones for LTL. 
A formula $A_{LTL}$ is an LTL {\em boxed formula} if it is 
%in the form of 
a sequence of boxes followed by an atom, where each element of the sequence belongs to $\{X, G\}$. A formula is an LTL {\em positive} formula if it can be constructed from all logical symbols and modal operators of LTL except negation; a formula $N_{LTL}$ is an LTL {\em negative} formula if it is the negation of an LTL positive formula.
%\fb{this is a bit of a repetition, as we've just defined boxed and positive/negative formulas. Also, it is not entirely correct for negative formulas.}

We now define LTL Sahlqvist formulas. 
%
%\fb{why defining Sahlqvist formulas in this way?}
%
\begin{definition}[LTL Sahlqvist Formulas] \label{defSahlqvist}
Suppose $\beta$ is an LTL boxed formula or negative formula. Then we define LTL {\em untied formula} as follows: 
\begin{eqnarray*}   
    \phi & = & A_{LTL} \mid N_{LTL} \mid \beta U \phi \mid \phi \land \phi 
\end{eqnarray*} 

The LTL {\em Sahlqvist formulas} are the conjunction of negations of LTL untied formulas. 
\end{definition}
\begin{remark}
    In the definition of LTL untied formula, $F\phi$ can be retrieved using $\top U \phi$.
\end{remark}
\vspace{0.2cm}

As for LTL', its Sahlqvist formulas are defined as follows: 
\begin{definition}[LTL' Sahlqvist Formulas]
    An LTL' untied formula is constructed from LTL' boxed formulas and LTL' negative formulas using only $F_x$ and conjunction: 
%    \cbstart
    \begin{eqnarray*}
        \phi  & =  &  A_{LTL'} \mid N_{LTL'} \mid \phi\wedge\phi \mid F_x\phi
    \end{eqnarray*}
%    \cbend
    %\fb{what is this notation? it is not really standard.} \rui{ It is just to make the construction more visual, maybe we can replace it by BNF?} \fb{yep, BNF would be better.}
    
    %The occurrence of $\phi\wedge\phi$ and $F_x\phi$ can be possibly empty. That is, negative formulas and strongly positive formulas alone are also considered untied. 
    
    As before, LTL' Sahlqvist formulas are the conjunctions of negations of LTL' untied formulas.
\end{definition}

\section{Correspondence Theorem} \label{correspondence}

In this section we present the proof of the correspondence theorem for LTL. By embedding LTL Sahlqvist formulas into LTL' Sahlqvist formulas, we only need to show that the theorem holds for the latter. We start by showing that the translation $t$ from LTL to LTL' in Sec.~\ref{semantics} preserves Sahlqvist formulas. 
%translation from LTL to LTL'. 
Then we introduce the main lemma crucial to the theorem. Finally, a detailed proof of the theorem is provided. 

\subsection{Translation}

We show that LTL Sahlqvist formulas can be translated into LTL' Sahlqvist formulas. 

%\cbstart
\begin{lemma}
    Let $\tau$ be the translation from LTL to LTL' in Def.~\ref{translation}. Then the following claims are true: 
    \begin{enumerate}[$(1)$]
        \item The translation of an LTL untied formula w.r.t.~$\tau$ is an LTL' untied formula. 
        \item An LTL untied formula and its translation w.r.t.~$\tau$ are semantically equivalent. %\fb{this should immediately follow from the preservation result in Sec.~\ref{semantics}.}
    \end{enumerate}
\end{lemma}
\begin{proof}
    
    \begin{enumerate}[$(1)$]
    \item The claim can be proved using structural induction on the formula. We only consider the case for the LTL untied formula $\phi = \beta U \psi$, where $\psi$ is also LTL untied. Let $w$ be any path, and $\phi$ is evaluated at $w$. By definition 2.4, $\tau(\phi)=F_x (\tau(\psi) \land \widehat{G}_{w,x} \tau(\beta))$. If $\beta$ is an LTL boxed formula, then $\tau(\beta)$ is also an LTL' boxed formula; so $\widehat{G}_{w,x} \tau(\beta)$ is also an LTL' boxed formula. If $\beta$ is an LTL negative formula, then $\tau(\beta)$ is an LTL' negative formula; so $\widehat{G}_{w,x} \tau(\beta)$ is also an LTL' negative formula. Therefore, $\widehat{G}_{w,x} \tau(\beta)$ is untied. By induction hypothesis, $\tau(\psi)$ is an LTL' untied formula, hence $F_x (\tau(\psi) \land \widehat{G}_{w,x} \tau(\beta))$ is LTL' untied. 
    
    %\item  \fb{I think we can remove this here and prove preservation in Sec.~\ref{semantics}. Also, it is better to prove semantical equivalence using the semantics rather than the standard translation.}
    
    \item It follows immediately from Lemma~\ref{translation_lemma}.
        \end{enumerate}
\end{proof}
%\cbend

Whenever two formulas are semantically equivalent, they have the same frame conditions. Therefore, having shown that for each LTL Sahlqvist formula, a semantically equivalent LTL' formula exists and is also Sahlqvist, we can conclude the following lemma:  
\begin{lemma}
    If every LTL' Sahlqvist formula locally corresponds to a first order formula, then every LTL Sahlqvist formula locally corresponds to a first order formula. 
\end{lemma}

\subsection{Main Lemma}

In this section, we prove the main lemma, essential to the proof of the
correspondence theorem for LTL'. The LTL' untied formulas are solely built from boxed formula and negative formula, hence intuitively in order to find first-order correpondents for LTL' Sahlqvist formula $\phi$, it suffices to find an assignment $h_0$ that satisfies the following for every boxed formula $A$ and every negative formula $N$ in $\phi$:
%\fb{we could remind what $A$ and $N$ are.}
\begin{eqnarray*}
\exists Q, ST_w(A) & \iff & ST_w(A)[Q_0] \text{ and }\\
 \exists Q, ST_w(N) & \iff & ST_w(N)[Q_0]
\end{eqnarray*}
where $Q_0(x)$ is true iff $x\in h_0(q)$. $Q_0$ is called $\textit{minimal predicate}$. 

\begin{definition}[Substitution]
    We first fix the notation on substitution in the minimal assignment. Let $\phi(q)$ be a formula and $h_0(q)$ be its minimal assignment for atom $q$ (to be defined subsequently). Let $Q_0$ be its corresponding minimal predicate. Suppose $t$ to be a symbol occurring in the expression of $Q_0$. Then we use $[t'/t]Q_0$ to denote the substitution of $t'$ for all occurrences of $t$ in $Q_0$.
\end{definition}
%\cbend

We can now introduce the notion of minimal assignment.
\begin{definition}[Minimal assignment]
    Let $\phi(q_1,\ldots, q_k)$ be an LTL' untied formula; let $w$ be a path. For every variable $q_j$ occurring in $\phi$, we define the {\em minimal assignment} $h_0(q_j)$ of $\phi$ at $w$ by induction on the structure of formula. 
%    \vspace{0.2cm}
%    \noindent
    
    \textit{Base cases: }
    Suppose that  $\phi (q_j)$ is a boxed formula and its standard translation at $w$ is $\forall v (R_j(w,v) \to Q_j(v))$, then the minimal assignment for $q_j$ is $h_0(q_j) = \{u \in W \mid R_j(w,u)\}$. 
    
    Suppose $\phi$ is a negative formula, then $h_0(q_j) = \emptyset$ (and $Q_{j0}(w) \equiv \bot$ for every $w$).
    \vspace{0.2cm}
    
    \noindent
    \textit{Inductive cases:}
    
    If the minimal assignment for $\phi_1(q_1,\ldots, q_k)$ and $\phi_2(q_1,\ldots, q_k)$ are respectively $h_0^1$ and $h_0^2$, then the minimal assignment for $\phi_1\wedge\phi_2$ is $h_0^1 \cup h_0^2$. 
    
%    \cbstart
    If the minimal assignment for $\phi$ at $v$ is $h_0$, then the minimal assignment for $F_x\phi$ at $w$ is $[x/v]h_0$.
    %$h'_0 = (h_0 \setminus \{ v \}) \cup \{ x \}$ \fb{be careful with notation}. 

    If the minimal assignment for $\phi$ at $v$ is $h_0$, then the minimal assignment for $X\phi$ at $w$ is $[\textbf{S}(w)/v]h_0$. %\fb{again, be careful with notation.}
    
    Suppose the minimal assignment for $\phi$ at $v$ is $h_0$. The minimal predicates for $q_j$ occurring in $\phi$ is defined as $Q_{j0}(z) \Leftrightarrow z\in h_0(q_j)$. Then the minimal assignment $h_0'$ for $\widehat{G}_{s,s'}\phi$ at $w$ is defined as $h_0'(q_j) = \{y\in W \mid \exists x (s\leqslant x < s' \wedge [x/v]Q_{j0}(y))\}$ for every $q_j$. 
%    \cbend
\end{definition}
\begin{remark}
    The minimal assignment for an LTL untied formula can be obtained by translating it into an LTL' untied formula. 
\end{remark}

Let $A$ be of the form $\boxplus^n q$ and $w \vDash A$ iff $\forall x (R^n(w,x) \to Q(x))$. Let $Q_{0}(x)$ be $R^n(w,x)$, we claim that $\exists Q \forall x (R^n(w,x) \to Q(x))$ iff $\forall x (R^n(w,x) \to Q_{0}(x))$. The proof of this claim is immediate: the right hand side is always true; the right-to-left implication is also always true. %Let $N(q)$ be a negative formula and $S(Q)$ be the standard translation of $N[w]$. Then $\exists Q (\forall x (R^n(w,x) \to Q(x))\wedge S(Q))$ iff $\forall x (R^n(w,x) \to Q_0(x))\wedge S(Q_0)$ by Corollary~\ref{conversemonotonicity}. So 
%\cbstart 
It turns out that for every Sahlqvist formula, the recursive construction of the minimal assignment will always produce a first-order correspondent to its second-order translation. In particular, we need to show how the occurrences of negative formulas in a Sahlqvist formula can be given such first-order correspondents via minimal assignment. 
%\cbend
%That is the Sahlqvist correspondence theorem. 
%\vspace{0.2cm}

\begin{lemma}[Main Lemma] \label{main}
%    \cbstart 
    Let $E$ be an LTL' untied formula, $w$ is a state, and $h_0$ is the minimal assignment of $E$ at $w$ (possibly empty). Let $h$ be an assignment. If there exists an assignment $g$ and a state $w$ such that $[E]^g_w = 1$, then the following are equivalent: 
    \begin{enumerate}[$(a)$]
        \item For all $q_j \in \{q_1,\ldots, q_k\}$, $h_0(q_j) \subseteq h(q_j)$.
        \item $[B]^h_w = 1$.
    \end{enumerate}
    where $B$ is defined is obtained from $E$ by replacing all occurrences of negative formulas $N_1, N_2,\ldots, N_m$ in $E$ by $\top$. 
%    \cbend
\end{lemma}

\begin{proof}
    We proceed by induction on the structure of the formula. 
    
    For the base cases, we suppose that $E$ is either an LTL' boxed formula $A$ or an LTL' negative formula $N$. If $E$ is a negative formula $N$, then $h_0$ is empty, therefore ($a$) must be true. As $B$ becomes $\top$, ($b$) is true, hence ($a$) and ($b$) are  equivalent. If $E$ is a boxed formula $A$, then $B = A$. As only one atom appears in $A$, let it be $q$. Since $A(q)$ is true at $w$, $Q_0(x)$ is $R(w,x)$ where $R$ is obtained from the standard translation of $A(q)$. As $h(q) = \{x \in W \mid Q(x)\}$ and $h_0(q) = \{x \in W \mid Q_0(x)\} = \{x \in W \mid R(w,x)\}$, $h_0(q) \subseteq h(q)$ is therefore just saying that for all $x$, $(R(w,x) \to Q(x))$. But this is exactly what ($b$) says. Namely, $[E]^h_w = 1$ iff $w \vDash A$ iff $\forall x (R(w,x) \to Q(x))$. Therefore ($a$) and ($b$) are equivalent.
    \vspace{0.2cm}
    
    There are two cases for inductive steps: $E_1\wedge E_2$, $F_xE$. 
    \vspace{0.2cm}
    
    \noindent
    \textit{Case $(i)$:} Suppose that there is an assignment $g$ making $E = E_1\wedge E_2$ true at $w$. Then $g$ also makes both $E_1$ and $E_2$ true at $w$. By induction hypothesis, $(a)\Leftrightarrow (b)$ holds for both $B_1$ and $B_2$. Let $h$ be an arbitrary assignment. Let $h_0^1, h_0^2, h_0$ denote the minimal assignments for $E_1, E_2, E$. We know $h_0^1, h_0^2 \subseteq h_0$. Also, for every atomic formula $q_j$ in $E$, it must be in either $E_1$ or $E_2$. Thus if for every $q_j$, $h_0(q_j)\subseteq h(q_j)$, then $[B_1]^h_w = [B_2]^h_w = 1$. So $(a)\Rightarrow (b)$ holds for $B$. Now assume $[B]^h_w=1$. As before, $[B_1]^h_w=1$ and $[B_2]^h_w=1$. So for every atom $q_j$, if $q_j$ occurs in $B_i$, then $h_0^i(q_j)\subseteq h(q_j)$ ($i\in\{1,2\}$). By definition of $h_0$, it is also the case that $h_0^1(q_j)\cup h_0^2(q_j) = h_0(q_j)\subseteq h(q_j)$ for every $q_j$. Hence $(b)\Rightarrow (a)$ holds for $B$. 
    \vspace{0.2cm}
    
    \noindent
    \textit{Case $(ii)$:} Suppose that there is an assignment $g$ making $F_xE$ true at $w$. Then there is a state $x \geqslant w$ at which $g$ makes $E$ true. By induction hypothesis, $(1)$ and $(2)$ are equivalent:
    \begin{enumerate}[$(1)$]
        \item For every atomic variable $q_j$ in $E$, $h_0(q_j) \subseteq h(q_j)$.
        \item $[B]^h_x=1$. 
    \end{enumerate}
    If $[B]^h_x=1$, then $[F_xB]^h_w=1$; so $(1)\Rightarrow [F_xB]^h_w=1$. Also, if $(a)\Rightarrow [F_xB]^h_w=1$, then there is a state $y$ at which $[B]^h_y=1$. Since $x$ in $(1)\Leftrightarrow (2)$ is arbitrary, we get $(1)$ back. Therefore, $(a)\Leftrightarrow (b)$ holds for $F_xB$.
%    \vspace{0.3cm}
\end{proof}

\subsection{Correspondence Theorem}

It finally only remains to show that all LTL' Sahlqvist formulas $S$ have first-order correspondents. 
\begin{theorem}[LTL Correspondence Theorem]
    Let $S$ be an LTL' Sahlqvist formula, then the local correspondent of $S[w]$ can be expressed in first-order terms, i.e., $\forall Q_1,\ldots, \forall Q_k, ST_w(S(q_1,\ldots, q_k))$ has a first-order correspondent.
\end{theorem}
\begin{proof}

Let $S = \bigwedge_{i=1}^m \neg E_i$ where $E_i$ are LTL' untied formulas. %\fb{we can define Sahlqvist formulas without conjunction} \rui{I added a note saying that we can work with each conjunctive clause separately}. 
The second order standard translation of $S[x]$ is $\bigwedge_{i=1}^m \forall Q_1, \forall Q_2,\ldots, \forall Q_k, \neg ST_x(E_i)$. However, to simplify the task, we can work with each conjunctive clause $E_i$ individually. In addition, we are going to work with the first correspondence of its negation: 
\begin{equation*}
    \exists Q_1, \exists Q_2,\ldots, \exists Q_k, ST_x(E_i)
\end{equation*}
We proceed by induction on the complexity of the formula. 
%\vspace{0.2cm}

\noindent
\textit{Base case: } Let's write the formula for the base case as follows:
\begin{equation*}
    \bigwedge\limits_{j=1}^m \oplus_{a_j}C_j 
\end{equation*}
where $C_j$ is either an LTL' boxed formula or an LTL' negative formula, and  $a_j$ is the number of $F_x$ appearing in front of each $C_j$.

For each $j$, if $C_j$ is a boxed formula, then the standard translation of $\oplus_{a_j}C_j[x]$ can be written as
\begin{equation*}
    \exists x_1, \ldots, \exists x_{a_j} ( R_{j1}(x,x_{j1}) \wedge \ldots \wedge R_{ja_j}(x_{a_j-1},x_{a_j}) \wedge (\forall y (R_j(x_{a_j},y) \to Q_j(y))))
\end{equation*}

However, if $C_j$ is a negative formula, we do not need to write down the standard translation of $\oplus_{a_j}C_j[x]$. We can omit $\oplus_{a_j}$ because it can be part of the LTL' negative formula. Therefore, the standard translation of $E_i[x]$ can be written as
\begin{equation} \label{eq11}
%\begin{split}
    \bigwedge\limits_{f=1}^t (\exists x_1, \ldots, \exists x_{a_f} \bigwedge\limits_{s=1}^{a_f}R_{fs}(x_{s-1},x_s) \wedge (\forall y (R_f(x_{a_f},y) \to Q_f(y)))) \wedge \bigwedge\limits_{l=1}^r ST_x(N_l)
%\end{split}
\end{equation}
where $t+r = m$. Here, $C_1,\ldots, C_m$ are $A_1,\ldots, A_t, N_1,\ldots, N_r$.

For the first conjunct of this formula, the following two formulas are equivalent by definition of minimal assignments, where $Q_{f0}$ is the minimal predicate of the atomic variable $Q_f$:
%\cbstart
{\small
\begin{eqnarray}
  & &   \exists Q_1, \ldots, Q_k \bigwedge\limits_{f=1}^t (\exists x_1, \ldots, x_{a_f} \bigwedge\limits_{s=1}^{a_f}R_{fs}(x_{s-1},x_s) \wedge (\forall y  (R_f(x_{a_f},y) \to Q_f(y))))  \ \ \ \  \ \label{eq3}\\
    %\label{eq3} & &\\
& &     \bigwedge\limits_{f=1}^t (\exists x_1, \ldots, \exists x_{a_f} \bigwedge\limits_{s=1}^{a_f}R_{fs}(x_{s-1},x_s) \wedge (\forall y (R_f(x_{a_f},y) \to Q_{f0}(y)))) \label{eq4} 
    %\label{eq4} & &
\end{eqnarray}
}
%\cbend

For the second conjunct of (\ref{eq11}) $\bigwedge_{l=1}^r ST_x(N_l)$, notice that $E_i$ satisfies the condition of the main lemma: namely, there exists an assignment under which it is satisfied at $x$. Also, let $h$ be an arbitrary assignment, if $[E_i]^h_x = 1$, then $[B_i]^h_x = 1$, where $B_i$ is obtained from $E_i$ by substituting $\top$ for every occurrence of negative formulas in $E_i$. It follows that for all $q_j \in \{q_1,\ldots, q_k\}$, $h_0(q_j) \subseteq h(q_j)$. As $N_l$ are negative formulas, by the monotonocity lemma for negative formulas (Corollary~\ref{conversemonotonicity}), $h(N_l) \subseteq h_0(N_l)$. Therefore, for all $h$, if $[N_l]^h_x=1$, then $[N_l]^{h_0}_x=1$. From this, we can easily prove that (\ref{eq1}) and (\ref{eq2}) below are equivalent. 
\begin{eqnarray}
    \exists Q_1, \exists Q_2,\ldots, \exists Q_k, \bigwedge\limits_{l=1}^r ST_x(N_l) \label{eq1}\\
%\end{}
%\begin{equation}
    \bigwedge\limits_{l=1}^r ST_x(N_l)[Q_{10}, Q_{20},\ldots, Q_{k0}] \label{eq2}
\end{eqnarray}

From the equivalence (\ref{eq3}) $\wedge$ (\ref{eq1})$\iff$(\ref{eq4}) $\wedge$ (\ref{eq2}),  %\fb{please use labels and refs, as otherwise these references will change and make the proof unintelligible.}, 
$(1)$ obtains its first order correspondent by substituting minimal predicate $Q_{10},\ldots, Q_{k0}$ for $Q_1,\ldots, Q_k$, hence the quantifiers over them can also be dropped. Therefore, $\forall Q_1,\ldots, \forall Q_k, ST_x(S) \equiv \bigwedge_{i=1}^m \neg (\exists Q_1,\ldots, \exists Q_k, ST_x(E_i))$ also has first-order correspondent $\bigwedge_{i=1}^m \neg ST_x(E_i)[Q_{10}/Q_1, \ldots, Q_{k0}/Q_k]$. %where $\phi[Q/Q']$ means substituting $Q$ for $Q'$ in $\phi$.
\vspace{0.2cm}

Now we proceed to the inductive steps. %As in Lemma~\ref{main}, t
There are two cases: 

%\vspace{0.2cm}
\noindent
\textit{Case 1: } Suppose the untied formula $E$ is of the form $F_yC$, where $C$ is an untied formula. If $E$ is true at the state $x$, then there is a state $y$ such that 
%there is 
$x\leqslant y$ and $C$ is true at $y$. By induction hypothesis, we can find the minimal predicates for $C$ which is $\Vec{Q_0^C} = Q_{10},\ldots, Q_{k0}$ such that 
$$\exists \Vec{Q} ST_y(C) \Leftrightarrow [\Vec{Q_0^C}/\Vec{Q}] ST_y(C)$$
Since the minimal predicate for $E$ is $Q_0^E = [x/y]\Vec{Q_0^C}$, we get
$$\exists \Vec{Q} ST_x(E) \Leftrightarrow [\Vec{Q_0^E}/\Vec{Q}] ST_x(E)$$
%\vspace{0.3cm}

\noindent
\textit{Case 2: } Suppose $E = E_1\wedge E_2$ where $E_1$ and $E_2$ are both untied. Let $q_1,\ldots,q_k$ be the atoms appearing in both $E_1$ and $E_2$. Then by induction hypothesis, we have two sets of minimal predicates $\{Q_{10}^1,\ldots, Q_{k0}^1\}$ and $\{Q_{10}^2,\ldots, Q_{k0}^2\}$. The minimal predicates for $E$ are defined as 
$$Q_{j0} = Q_{j0}^1\vee Q_{j0}^2$$
By induction hypothesis, we know that 
$$\exists \Vec{Q} ST_x(E_1) \Leftrightarrow [\Vec{Q_0^1}/\Vec{Q}] ST_x(E_1)$$
$$\exists \Vec{Q} ST_x(E_2) \Leftrightarrow [\Vec{Q_0^2}/\Vec{Q}] ST_x(E_2)$$
We want to show that
$$\exists \Vec{Q} ST_x(E) \Leftrightarrow [\Vec{Q_0}/\Vec{Q}] ST_x(E)$$

$(\Rightarrow):$ Assume $\exists \Vec{Q} ST_x(E)$. Then, $\exists \Vec{Q} ST_x(E_1)$ and $\exists \Vec{Q} ST_x(E_2)$ hold. So both $[\Vec{Q_0^1}/\Vec{Q}] ST_x(E_1)$ and $[\Vec{Q_0^2}/\Vec{Q}] ST_x(E_2)$ are the case. Since $h_0^1\subseteq h_0$, $[\Vec{Q_0}/\Vec{Q}] ST_x(E_1)$ is also true. Similarly, so is $[\Vec{Q_0}/\Vec{Q}] ST_x(E_2)$. Therefore, $[\Vec{Q_0}/\Vec{Q}] ST_x(E)$.

%$(\Leftarrow):$ Suppose $[\Vec{Q_0}/\Vec{Q}] ST_x(E)$. Then, both $[\Vec{Q_0}/\Vec{Q}] ST_x(E_1)$ and $[\Vec{Q_0}/\Vec{Q}] ST_x(E_2)$ are the case. Since $h_0^1 \subseteq h_0$ and $h_0^1$ is minimal for $E_1$, by Lemma~\ref{main}, $[\Vec{Q_0^1}/\Vec{Q}] ST_x(E_1)$ is true. So $\exists \Vec{Q} ST_x(E_1)$. Similarly, $\exists \Vec{Q} ST_x(E_2)$ holds. Hence, $\exists \Vec{Q} ST_x(E_1\wedge E_2)$.

$(\Leftarrow):$ $\Vec{Q_0}$ is an instance of $\Vec{Q}$. 
%Suppose $[\Vec{Q_0}/\Vec{Q}] ST_x(E)$. Then, both $[\Vec{Q_0}/\Vec{Q}] ST_x(E_1)$ and $[\Vec{Q_0}/\Vec{Q}] ST_x(E_2)$ are the case. Since $h_0^1 \subseteq h_0$ and $h_0^1$ is minimal for $E_1$, by Lemma~\ref{main}, $[\Vec{Q_0^1}/\Vec{Q}] ST_x(B_1)$ is true, where $B_1$ is obtained from $E_1$ by substituting $\top$ for all occurrences of negative formulas $N_1,\ldots,N_m$ occurring in $E_1$. Moreover, by Corollary~\ref{conversemonotonicity}, if there is an assignment that makes $N_1,\ldots,N_m$ true at some states, then the minimal assignment $h_0^1$ also makes them true at the same states; so $[\Vec{Q_0^1}/\Vec{Q}] ST_x(E_1)$ is true. Therefore, $\exists \Vec{Q} ST_x(E_1)$ holds. Similarly, $\exists \Vec{Q} ST_x(E_2)$ also holds. Hence, $\exists \Vec{Q} ST_x(E_1\wedge E_2)$.

\end{proof}

%\begin{remark}
 %   The same correspondence proof can be applied to Sahlqvist formulas of standard modal logics. \fb{is this really informative?} 
%\end{remark}

This concludes the proof of the Sahlqvist correspondence theorem for LTL'. First-order correspondents for LTL can be found by first translating the LTL Sahlqvist formulas into LTL'.

%\section{Descriptive assignment}
%\section{Intersection Lemma}
%\section{Persistence theorem for LTL Sahlqvist formula}

\subsection{Example}

The above proof of the correspondence theorem also yields an algorithm for translating the frame condition of an LTL Sahlqvist formula into a first order formula. We do not elaborate the algorithm here. But the algorithm for the Sahlqvist formula for standard modal logic applies with appropriate modification. Let's see an example for LTL involving the Until operator. 
%\cbstart 
Let $\varphi = \neg (\neg q U q)$, readers can easily verify that it is an LTL Sahlqvist formula. 
%\cbend 
The standard translation of $\varphi [w]$ is 
$$\neg\exists Q(\exists v, w \leqslant v \wedge (\forall u, v\leqslant u \to Q(u))\wedge (\forall z, w\leqslant z < v \to \neg Q(z)))$$
Taking the minimal assignment $Q(x)\equiv v\leqslant x$, we reduce the $ST_w(\varphi)$ to 
\begin{equation} \label{example}
\neg (\exists v, w \leqslant v)
\end{equation}
Formula ($\ref{example}$) identifies the empty class of structures, as there exists no class of frames over which formula ($\ref{example}$) can be true at any state. 
%\rui{Is this last claim correct?}

\section{Conclusions}

In this paper we introduced a notion of Sahlqvist formula for the Linear-time Temporal Logic LTL and proved a 
%This paper is only about the Sahlqvist formulas and 
Sahlqvist correspondence theorem for this language. In some respects, they can be viewed as a generalization of the same result for standard modal logic, in the sense that we allow states to index temporal operators $F_x$ and $G_{x,x'}$. 
%\fb{what does it means?}. 
One should also remark that LTL' Sahlqvist formulas are in fact very similar to the Sahlqvist formulas of standard modal logic to the extent that the proof for the completeness property \cite{BlackburnRV02,GabbayHR94} for Sahlqvist formulas almost identically applies to the LTL' Sahlqvist formulas.

%\fb{please fix the missing references.}

Further research direction may consists in finding an even larger class of LTL Sahlqvist formulas. For standard modal logic, Chagrova \cite{Chagrova91} has proved that it is undecidable if an arbitrary formula has a first-order correspondent. Therefore, the same problem is equally undecidable for LTL as the latter is strictly more expressive than the former. 

\newpage

%\bibliographystyle{aiml22}
%\bibliography{aiml22}

\end{document}